\documentclass[a4paper,UKenglish,cleveref, autoref, thm-restate]{lipics-v2021}

\usepackage{dsfont}
\usepackage{xspace}
\usepackage{wrapfig}

\newcommand{\A}{{\cal A}}

\newcommand{\C}{{\cal C}}

\newcommand{\G}{{\cal G}}

\renewcommand{\H}{{\cal H}}

\renewcommand{\O}{{\cal O}}

\newcommand{\Th}{\text{Th}}

\newcommand{\oalpha}{{\overline{\alpha}}}

\newcommand{\play}{\text{play}}
\renewcommand{\path}{\text{path}}
\newcommand{\util}{\text{util}}
\newcommand{\MP}{\text{MP}}
\newcommand{\CFG}{\text{CFG}}

\newcommand{\PZ}{Player~$0$\xspace}
\newcommand{\PO}{Player~$1$\xspace}
\newcommand{\PT}{Player~$2$\xspace}
\newcommand{\PR}{Player~$3$\xspace}

\newcommand{\PLi}{Player~$i$\xspace}
\newcommand{\PLj}{Player~$j$\xspace}
\newcommand{\CoaX}{Coalition~$X$\xspace}
\newcommand{\CoaY}{Coalition~$Y$\xspace}
\newcommand{\CoaA}{Coalition~$\alpha$\xspace}

\renewcommand{\CFG}{\text{CFG}}

\usepackage{xcolor}

\newcommand{\Nat}{\mathds{N}}

\newcommand{\Z}{\mathbb{Z}}

\newcommand{\zug}[1]{\langle #1  \rangle}
\newcommand{\set}[1]{\{ #1 \}}

\newcommand{\stam}[1]{}
\newcommand{\short}[1]{}

\usepackage{tikz}
\usepackage{forest}

\usepackage{amsmath}
\usepackage{amssymb}
\usepackage{diagbox} 
\usetikzlibrary{automata, positioning, arrows.meta}

\title{Multi-Player Discrete-Bidding Games; Determinacy, Equilibria, and Complexity} 

\author{Guy Avni}{Department of Computer Science, University of Haifa, Israel \and \url{https://sites.google.com/view/gavni/home} }{gavni@cs.haifa.ac.il}{}{}

\author{Fatima Murra}{Department of Computer Science, University of Haifa, Israel}{f.m.murra@gmail.com}{}{}

\authorrunning{G. Avni and F. Murra} 

\Copyright{Guy Avni and Fatima Murra} 



\category{} 

\relatedversion{} 

\usetikzlibrary{arrows.meta, positioning, automata}

\EventEditors{John Q. Open and Joan R. Access}
\EventNoEds{2}
\EventLongTitle{42nd Conference on Very Important Topics (CVIT 2016)}
\EventShortTitle{CVIT 2016}
\EventAcronym{CVIT}
\EventYear{2016}
\EventDate{December 24--27, 2016}
\EventLocation{Little Whinging, United Kingdom}
\EventLogo{}
\SeriesVolume{42}
\ArticleNo{23}

\begin{document}

\maketitle

\begin{abstract}
Games on graphs constitute a fundamental model. Applications include {\em reactive synthesis}, which reduces to solving a zero-sum two-player game, and reasoning about multi-agent systems by modeling them as a multi-player game. We study a class of graph games called {\em bidding games} in which the players are allocated a budget, and in each turn, an auction determines which player moves the token. Two-player bidding games have been extensively studied. We study, for the first time, multi-player bidding games. We focus on discrete bidding, which imposes granularity restrictions on the players' budgets and bids. The original motivation for discrete bidding is practical applications, and technically, it is appealing that the game has only finitely-many configurations.
We initiate our study by considering a game between two coalitions of players. We show that under mild assumptions on the mechanism that is applied to break bidding ties, bidding games are {\em determined}: from every initial configuration, one of the coalitions has a winning strategy. Thus, we identify a sub-class of multi-player concurrent games that is determined. We extend the result to existence of a value in mean-payoff games. Using determinacy, we show that a pure Nash equilibrium always exists in qualitative games. 
Finally, we show that the complexity of deciding which coalition wins from a given configuration is PSPACE-hard already in reachability games and already for budgets given in unary. This is in stark contrast to two-player games, which are known to be in NP and coNP even for budgets given in binary.
\end{abstract}

\section{Introduction}
\label{sec:Introduction}
Graph games constitute a fundamental model. Theoretically, two-player zero-sum games have a deep connection to foundations of logic~\cite{Rab69} and have favorable properties; in particular, {\em turn-based} games, i.e., the players take actions alternatively, with {\em Borel} objectives are {\em determined}~\cite{Mar75}: in each game, from every position, one of the players has a winning strategy. 
In terms of applications, {\em reactive synthesis}~\cite{PR89} in which the goal is to develop a reactive system (also called a {\em controller}) that satisfies its objective in any environment, naturally reduces to the problem of finding a winning strategy for a system player against an adversarial environment. 
Taking a step further, modeling the system and environment as monolithic entities is often too crude. {\em Multi-player} graph games enable a refined modeling of a system and environment that consist of interacting components. Then, following the seminal~\cite{AHK02}, one asks: {\em can a coalition of players (e.g., the system) guarantee their objective (e.g., satisfy a specification) no matter how a coalition of the rest of the players (e.g., the environment) acts}. 
Yet another refinement arises by observing that each component has its individual objective, called {\em rational synthesis}~\cite{FKL10,KPV16,WG+16}. This interaction is captured by multi-player {\em non-zero-sum} games. {\em Nash equilibrium}~\cite{Nas50} (NE) is the basic {\em solution concept} in  games in general and graph games in particular. While Nash's theorem establishes existence of NE in {\em mixed} strategies, i.e., strategies that randomize over actions, in turn-based graph games, a NE in {\em pure} strategies always exists~\cite{Ume11}. 

\paragraph*{Bidding games}
This paper focuses on {\em bidding games}~\cite{LLPSU99,AHC19}.
So far, only two-player bidding games have been studied. We study, for the first time, multi-player bidding games. 


A bidding game is played on a graph as follows. A token is placed on an initial vertex. 
Each player has a budget and in each turn, an auction (bidding) determines which player moves the token. This generates an infinite path that determines the winner or payoff of the game. Concretely, suppose that there are $m \in \Nat$ players and the budgets are $B_1,\ldots, B_m$. The players simultaneously choose bids that do not exceed their budgets, namely $b_i \leq B_i$, for $1 \leq i \leq m$, the highest bidder moves the token and pays his bid evenly among the other players. This is called {\em Richman} rules (see other rules in~\cite{AH26}). 
We distinguish between {\em continuous} and {\em discrete} bidding. The latter is the focus of this paper and imposes a granularity restriction on the bids: we require both $B_i$ and $b_i$ to be natural numbers, and denote the total budget by $k = \sum_{1 \leq i \leq m} B_i$. 
We discuss bidding ties and rounding of payments below. 
Observe that $k$ stays constant throughout the game. 

\noindent{\bf Two-player bidding games.} We survey the relevant literature. 
The central question in bidding games regards the {\em threshold} budget; intuitively, a budget that is both necessary and sufficient to guarantee an objective. 
In continuous-bidding games, it was shown in~\cite{LLPSU99,AHC19} that each vertex $v$ has a threshold $\Th(v)$ such that, for every $\epsilon > 0$, if \PO's budget is $\Th(v) + \epsilon$, he has a winning strategy, and if \PO's budget is $\Th(v) - \epsilon$, \PT has a winning strategy. Moreover, $\Th(v)$ has the following structure, which we call the {\em average property}: let $v^-$ and $v^+$ be the neighbors of $v$ that respectively achieve the minimal and maximal thresholds, then $\Th(v) = 0.5 \cdot \big(\Th(v^-) + \Th(v^+)\big)$. In fact, the proof for {\em reachability games}, i.e., one of the players wins if a target vertex is reached, proceeds by establishing the average property in games on trees, then extending to general graphs~\cite{LLPSU99}. 

We turn to two-player discrete-bidding games, which have been studied using two techniques. (1)~Establishing a {\em discrete average property}. This mirrors the technique in continuous-bidding games; roughly, the threshold is a rounded average of the thresholds in two neighbors. 
(2)~{\em Local vs global determinacy}~\cite{AAH21}; reasons about the possible outcomes of a bidding and establishes that (locally) one of the players can force staying in a winning region, which leads to (global) determinacy. The framework was studied in concurrent games in general~\cite{BBR21}. 
The discrete-average property for reachability games as shown in~\cite{DP10}. For infinite-duration objectives, existence was shown using the second technique~\cite{AAH21}, but only recently the average property was shown for {\em parity} objectives~\cite{AS25} and for {\em mean-payoff} objectives~\cite{AS26}. 
Establishing the average property has advantages over the other technique; succinctly-represented optimal strategies can be constructed, which leads to improved complexity bounds, namely deciding the winner is in NP and coNP even when $k$ is given in binary (i.e., the game has exponential-many underlying {\em configurations}).

\noindent{\bf Multi-player bidding games.}
We start our study with the basic question on multi-player games, adapted from the seminal~\cite{AHK02}: {\em Can a coalition of players force winning the game?} 
In more detail, consider a configuration $c = \zug{v, B_1, \ldots, B_m}$, which means that the token is placed on $v$ and \PLi's budget is $B_i$, for $1 \leq i \leq m$. Suppose that the players $\set{1,\ldots,m}$ are partitioned into two coalitions $X$ and $Y$. We ask: can Coalition~$X$ guarantee an objective from $c$ no matter how Coalition~$Y$ plays? The following example questions the applicability of a technique based on thresholds in multi-player games. 
\begin{example}
\label{ex:no-thresholds}
Let $k = 11$, and $X = \set{1}$ and $Y = \set{2,3}$. Consider the game depicted in Fig.~\ref{fig:win-once}; intuitively, \CoaX needs to win one of the first two biddings. Fix \PO's budget to $B_1 = 3$. We show that winning depends on the allocation to the other players. On the one hand, \CoaY wins when the allocation is uniform; from $\zug{v_0, 3, 4, 4}$ \PT bids $b_2 = 4$ leading to $\zug{v_1, 5, 0, 6}$, then \PR bids $b_3 = 6$ to win the game. Second, \CoaX wins from $\zug{v_0, 3, 8, 0}$. Indeed, even if \PT wins the first bidding with $b_2 = 4$, the game reaches  $\zug{v_1, 5, 4, 2}$. But then \PO wins the bidding with $b_1 = 5$ and draws the game to $t$. A similar argument shows that \CoaX wins from $\zug{v_0, 3, 0, 8}$. Thus, the set of winning configurations for \CoaX is not convex!
\end{example}
\label{ex:intro1}
\begin{figure}[ht]
\begin{minipage}{0.5\textwidth}
\centering
\resizebox{!}{2cm}{%
\begin{tikzpicture}[
    >=Stealth,
    node distance=1cm and 2cm,
    state/.style={circle, draw, minimum size=2em, font=\large\bfseries},
    accepting/.style={double, double distance=2pt},
    thick,
]

    \node[state] (v0) at (0, 0) {$v_0$};
    \node[state] (v1) [right=1.75cm of v0] {$v_1$};
    \node[state] (s)  [right=1.75cm of v1] {$s$};
    \node[state] (t)  [right=1.75cm of s, accepting] {$t$};

    \draw[->] (-1, 0) -- (v0); 
    \draw[->] (v0) -- (v1);
    \draw[->] (v1) -- (s);
    \draw[->] (v1) to[out=35, in=180] (t);
    \draw[->] (v0) to[out=-35, in=185] (t); 
\end{tikzpicture}%
}
\caption{\CoaX's target is $t$ and \CoaY wins if $s$ is reached.}
\label{fig:win-once}
\end{minipage}%
\begin{minipage}{0.5\textwidth}%
    \centering
    \begin{tikzpicture}[
        scale=1.0,
        vertex/.style={circle, draw, minimum size=0.65cm, font=\small, thick},
        arrow/.style={->, >=Stealth, thick}
    ]

        \node[vertex] (v2) at (0,0) {$v_2$};   
        \node[vertex] (v1) at (1.5,0) {$v_1$};   
        \node[vertex] (v0) at (3.0,0) {$v_0$};   
        \node[vertex] (v_r) at (4.5,0) {$v_4$};  
        \node[vertex] (v3) at (6.0,0) {$v_3$};   

        \node[below=0.05cm of v0, font=\small] {$(1,2,3)$};
        \node[below=0.05cm of v3, font=\small] {$(2,3,1)$};
        \node[below=0.05cm of v2, font=\small] {$(3,2,1)$};

        \draw[arrow, thin] (3.0,0.8) -- (v0.north);

        \draw[arrow] (v0) -- (v1);
        \draw[arrow] (v0) -- (v_r);
        \draw[arrow] (v1) -- (v2);
        \draw[arrow] (v_r) -- (v3);

        \draw[arrow] (v2) to[bend left=30] (v3);
        \draw[arrow] (v3) to[bend left=30] (v2);

        \draw[arrow] (v2) to[out=135, in=225, looseness=4] (v2);
        \draw[arrow] (v3) to[out=45, in=315, looseness=4] (v3);

    \end{tikzpicture}
    \caption{A three player game.}
    \label{fig:ABS}
    \end{minipage}
\end{figure}

In fact, a variant of Ex.~\ref{ex:intro1} has been found for continuous bidding, and cut short any research on multi-player bidding games. Recent results on discrete-bidding games motivate us to approach this model under discrete-bidding. Furthermore, fortunately, unlike continuous-bidding games, a second reasoning technique is available for discrete-bidding games. Before presenting our results, we show that, as is the case in two-player games, determinacy holds only for certain tie-breaking mechanisms. 

\begin{example}
Consider the simple game depicted in Fig.~\ref{fig:example2} (a) with four players and a total budget of $k=4$. 
The tie-breaking mechanism is presented in Fig.~\ref{fig:example2} (b); the first entry in each row states the players that choose the highest bid, and the second entry states which player wins the bidding. 
Consider the initial configuration $c_0 = \zug{v_0, 1, 1, 1, 1}$ with coalitions $X = \set{1,2}$ and $Y = \set{3,4}$. 
\CoaX's target is $t$. 

We claim that neither coalition has a winning strategy. 
Fig.~\ref{fig:example2} (c) depicts the {\em bidding matrix} for configuration $c_0$; it captures all possible outcomes of the bidding at $c_0$. 
Each entry corresponds to a pair of joint coalition bids and indicates the coalition that wins the game. 
For example, when \CoaX bids $(0,1)$ and \CoaY bids $(1,0)$ a tie occurs between Players~$2$ and~$3$, \PT wins the tie, proceeds to $t$, and \CoaX wins the game. 
Observe that the bidding matrix neither has an $X$-row, a row all of whose entry are $X$, nor a $Y$-column, a column all of whose entries are $Y$, thus neither coalition has a winning strategy. 
\end{example}

\stam{
OLD
As shown in Figure 1, the game is not determined when non-linear tie-breaking is used. The game in Figure 1 consists of four players, with initial budget 1, divided into two coalitions, $A$ and $B$. The objective for Coalition $A$ is to visit vertex t infinitely often, while Coalition $B$ aims to prevent them from doing so. T is a tie-breaking mechanism that defined as: T(3,2)$\to$2, T(2,4)$\to$4, T(1,4)$\to$1, T(1,3)$\to$3, T(1,2,3,4)$\to$1. This game is not determined as we can see in Figure 2, for every possible strategy of Coalition $A$ from vertex $v$, there exists a winning strategy of Coalition $B$. The same holds true for the inverse case. It follows that neither coalition has a winning strategy from $v$.
\end{example}

We extend the local vs global determinacy technique to multi-player bidding games and establish determinacy. Thus, we provide positive results on the set of winning configurations, even though they are not necessarily convex. We also study mean-payoff multi-player bidding games and show that a value always exists. We study, for the first time, non-zero-sum multi-player bidding games and show existence of NE in pure strategies.

TODO: move the example to intro. In the definitions, we already use linear tie-breaking and this is a good motivation for the definitions. 


}

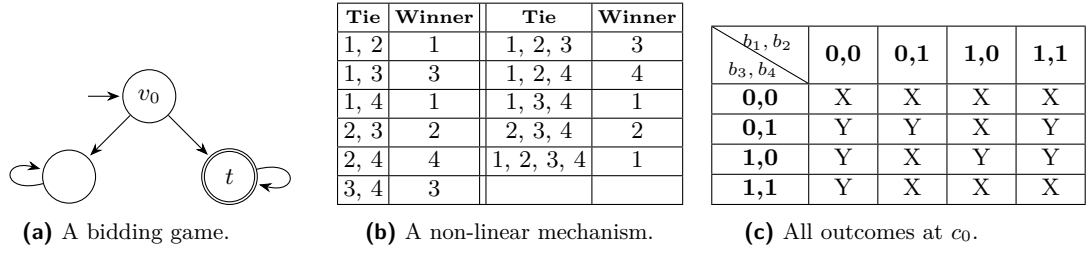
\begin{figure}[htbp]
  \centering
  
  \begin{subfigure}[b]{0.3\textwidth}
    \centering
    \captionsetup{width=0.8\textwidth}
    \begin{tikzpicture}[
      shorten >=1pt, 
      node distance=1.5, 
      on grid, 
      auto, 
      >=Stealth,
      every state/.style={draw=black, fill=white, minimum size=7mm, font=\small}
    ]
      \node[state, initial, initial where=above left, initial text=] (v) {$v_0$};
      \node[state] (left) [below left=of v] {};
      \node[state, accepting] (t) [below right=of v] {$t$};
      
      \path[->] 
        (v) edge node {} (left)
        (v) edge node {} (t)
        (left) edge [loop left] node {} (left)
        (t) edge [loop right] node {} (t);
    \end{tikzpicture}
    \caption{A bidding game.}
    \label{subfig:game}
  \end{subfigure}
  \hfill
  \begin{subfigure}[b]{0.34\textwidth}
    \centering
    \scriptsize 
    \setlength{\tabcolsep}{2pt} 
    \renewcommand{\arraystretch}{1.3} 
    \captionsetup{width=0.8\textwidth}
    \begin{tabular}{|c|c||c|c|}
      \hline
      \textbf{\scriptsize Tie} & \textbf{\scriptsize Winner} & \textbf{\scriptsize Tie} & \textbf{\scriptsize Winner} \\ \hline
      1, 2 & 1 & 1, 2, 3 & 3 \\ \hline
      1, 3 & 3 & 1, 2, 4 & 4 \\ \hline
      1, 4 & 1 & 1, 3, 4 & 1 \\ \hline
      2, 3 & 2 & 2, 3, 4 & 2 \\ \hline
      2, 4 & 4 & 1, 2, 3, 4 & 1 \\ \hline
      3, 4 & 3 & & \\ \hline 
    \end{tabular}
    \caption{A non-linear mechanism.} 
    \label{subfig:ties}
  \end{subfigure}
  \hfill
  \begin{subfigure}[b]{0.32\textwidth}
    \centering
    \small
    \renewcommand{\arraystretch}{0.95}
    \captionsetup{width=0.8\textwidth}
    \begin{tabular}{|c|c|c|c|c|}
      \hline
      \diagbox[width=3.9em, height=2.3em]{\scriptsize $b_3, b_4$}{\scriptsize $b_1, b_2$} & \textbf{\small 0,0} & \textbf{\small 0,1} & \textbf{\small 1,0} & \textbf{\small 1,1} \\ \hline
      \textbf{0,0} & X & X & X & X \\ \hline
      \textbf{0,1} & Y & Y & X & Y \\ \hline
      \textbf{1,0} & Y & X & Y & Y \\ \hline
      \textbf{1,1} & Y & X & X & X \\ \hline
    \end{tabular}
    \caption{All outcomes at $c_0$.}
    \label{subfig:matrix}
  \end{subfigure}

  \vspace{10pt}
  \caption{From configuration $c_0 = \zug{v_0, 1, 1, 1, 1}$, with coalitions $X = \set{1,2}$ and $Y = \set{3,4}$, and \CoaX's target is $t$, neither coalition has a winning strategy; there is no $X$-row nor a $Y$-column.}
  \label{fig:example2}
\end{figure}
 
\noindent{\bf Our results.}
We consider $m$-player games in which both tie-breaking and rounding mechanisms are {\em linear}: each vertex $v$ is associated with two orders $T_v$ and $R_v$, among the highest bidders, the player who is highest according to $T_v$, wins the bidding and assuming the winning bid is $b = (m-1) \cdot a + r$, then $a$ coins are paid to each player, and the remaining $r$ coins are given to the top $r$ players according to $R_v$. We adapt the framework of local to global determinacy to multi-player games to show that multi-player discrete-bidding games are determined: from every initial configuration, one of the two coalitions has a winning strategy. 

We study the complexity of deciding which coalition wins a from an initial configuration. We show that the problem is PSPACE-hard already for reachability objectives and even when the budgets are given in unary. This is in stark contrast to two-player games, where deciding the winner is in NP and coNP even when the budgets are given in binary~\cite{AS25}. That is, both games have exponential-many configurations. 
In fact, lower bounds are generally scarce in bidding games and quite weak: two-player discrete-bidding games for a class of objectives $\Gamma$ are only known to be harder than turn-based games with objective $\Gamma$: e.g., two-player reachability discrete-bidding games are only known to be P-hard.


We study, for the first time, non-zero-sum multi-player bidding games. Previously, zero-sum bidding games were only studied with two players and for games played on trees (and under continuous-bidding)~\cite{MKT18}. We show existence of Nash equilibrium with pure strategies. Finally, we initiate the study of multi-player mean-payoff games and show existence of a value in these games.

\paragraph*{Application; Auction-based scheduling.} 
We describe a direct application of our work in {\em decoupled multi-objective planning}. Given a graph $G$ and an objective $\varphi = \varphi_1\wedge \ldots \wedge \varphi_m$, the goal is to find a path in $G$ that satisfies $\varphi$. 
A decoupled approach proceeds by independently constructing a policy $f_i$ to each property $\varphi_i$, then composing $f_1,\ldots, f_m$ at runtime to produce a path that satisfies all. Advantages of decoupling include parallel construction of policies and reusability as we demonstrate below. 

The challenge is to design the policies as independently as possible while still providing a guarantee that their composition satisfies all objectives. 
In the recent~\cite{AH+26}, a scheduler randomly chooses a policy to act in each turn, and {\em almost-sure} satisfaction is guaranteed.  
In {\em auction-based scheduling} (ABS)~\cite{AMS24}, each policy is allocated a budget and in each turn, an auction among the policies determines which policy acts. 
Policy $f_i$ is constructed by finding a winning strategy in a bidding game with objective $\varphi_i$. Since so far, multi-player bidding games have not been considered, ABS was confined to decoupling of $m=2$ properties. Our work opens the door for decoupling $m >2$ policies as we demonstrate below. 

\begin{example}[Multi-objective auction-based scheduling]
Consider the graph that is depicted in Fig.~\ref{fig:ABS} and the specification $\varphi = \varphi_1 \wedge \varphi_2 \wedge \varphi_3$, with $\varphi_1 = F v_1$, $\varphi_2 = GF v_2$, and $\varphi_3 = GF v_3$. We apply ABS as follows. Set $k=3$, an initial configuration $c_0 = \zug{v_0, 1,1,1}$, the same rounding mechanism $R_v =(2,3,1)$ in every vertex $v$, and the tie-breaking mechanism is depicted below the vertices. Consider three zero-sum games: for $i = 1,2,3$, let $\G^i$ be a game between Coalition~$X = \set{i}$ and Coalition~$Y = \set{1,2,3} \setminus \set{i}$, where the goal of \CoaX is $\varphi_i$. 
We claim that \CoaX wins in all three games; let $f^X_i$ be a winning strategy in $\G^i$. In $\G^1$, $f^X_1$ bids $1$, wins the first bidding due to tie-breaking, and moves left to satisfy $\varphi_1$. For $i=2,3$, observe that due to rounding, \PO is never paid, and due to tie-breaking, \PLi cannot force staying at $v_i$. 
Then, the play that is generated by composing $f^X_1,f^X_2$, and $f^X_3$ is guaranteed to satisfy $\varphi$. 

We point to advantages. First, reusability: if $\varphi_1$ changes to $\varphi'_1 = F v_4$, i.e., the target moves, only $f^X_1$ needs to be change and the $f^X_2$ and $f^X_3$ can stay fixed. Second, applying different specialized algorithms to each task: in this case, $f^X_1$ was constructed by solving a reachability and $f^X_2$ and $f^X_3$ by solving a more complicated B\"uchi game. 
\end{example}

\section{Preliminaries}
We formalize the semantics of a bidding game by describing the explicit {\em concurrent game} that corresponds to it. We start by defining concurrent games. For $n \in \Nat$, we use the notation $[n] = \set{1,\ldots, n}$.

\paragraph*{Concurrent games}
A multi-player concurrent game is intuitively played as follows. A token is placed on an initial vertex, and in each turn, the players simultaneously select actions and the joint vector of actions determines the vertex that the token moves to. Formally, for $m \in \Nat$, an $m$-player concurrent game is played on arena $\zug{m, A,V, \lambda, \delta}$, where $A$ is a set of actions, $V$ is a set of vertices, $\lambda: V \times [m] \to 2^A$ assigns a set of actions that is allowed to each player in each vertex, and $\delta: V \times A^m \to V$ is a deterministic transition function that takes an action from each player and specifies the next vertex. For $v \in V$, we use $N(v)$ to denote the {\em neighbors} of $v$, defined as $N(v) = \set{v': \exists a \in A^m \text{ s.t. } \delta(v, a) = v'}$. 

We say that a player {\em controls} a vertex $v$ if their actions uniquely determine the successor vertex. For ease of presentation, restrict to \PO; he controls $v$ if for every $v' \in V$, there is an action $a_1$ such that for all $a_2,\ldots, a_m$ we have $\delta(v, (a_1,\ldots, a_m)) = v'$. A {\em turn-based} game is a special case of a concurrent game in which every vertex is controlled by a player. Note that a concurrent game could have vertices that are controlled by players and some that are not. 

\noindent{\bf Semantics.} A {\em history} records the history of vertices traversed and the actions chosen by the players. We use $\H$ to denote the set of histories, thus $\H = V \cdot (A^m \cdot V)^*$. For $i \in [m]$, a {\em strategy} for  \PLi is a function $\sigma_i: \H \rightarrow A$. A collection of strategies $\sigma_1,\ldots, \sigma_m$ and an initial vertex $v_0 \in V$ gives rise to a unique {\em play}, denoted $\play(v_0, \sigma_1,\ldots, \sigma_m) = v_0, \overline{a}_1, v_1, \overline{a}_2,\ldots$, which is defined inductively as follows. The initial vertex is $v_0$. Assume $h \in \H$ has been defined and ends in vertex $v_n$. Then, the next joint action is $\overline{a}_{n+1} = \zug{\sigma_1(h), \ldots, \sigma_m(h)}$ and the next vertex is $v_{n+1} = \delta(v_n, \overline{a}_{n+1})$. The {\em path} that the play traverses is $\path(v_0, \sigma_1,\ldots, \sigma_m) = v_0, v_1, \ldots$.
We restrict attention to strategies that select legal actions, namely $\sigma_i(h) \in \lambda(v, i)$, for every history $h$ that ends in $v$ and $i \in [m]$. 

We extend the definition to strategies for a {\em coalition} of players. Let $X \subseteq [m]$ be a coalition of players. A joint strategy $\sigma^X: \H \rightarrow A^X$ for Coalition~$X$ specifies an action for each player in the coalition following a history. We can think of $\sigma^X$ as a collection of $|X|$ strategies for the members of $X$ that {\em agree} with $\sigma^X$. 
Let $Y = [m] \setminus X$. We use $\play(v_0, \sigma^X, \sigma^Y)$ to denote the play that is obtained when the members of $X$ follows strategies that agree with $\sigma^X$ and the members of $Y$ follow strategies that agree with $\sigma^Y$. We define $\path(v_0, \sigma^X, \sigma^Y)$ similarly.

\paragraph*{Objectives and determinacy}
An {\em objective} determines which infinite paths are winning for a coalition $X \subseteq [m]$. Formally, it is $\O \subseteq V^\omega$ and we focus on the standard $\omega$-regular objectives. For a path $\pi \in V^\omega$, we use $\inf(\pi)$ to denote the vertices that $\pi$ visits infinitely often. 
\begin{description}
\item[Reachability.] There is a collection of {\em target} vertices $T \subseteq V$ and a path $\pi \in V^\omega$ is winning for Coalition~$X$ iff $\pi$ visits $T$. 
\item[Parity (``max odd'').] Each vertex is associated with a parity index, given by $p: V \rightarrow \set{1,\ldots,d}$. 
For a path $\pi=v_0,v_1,...$, denote the vertices that $\pi$ visits infinity often by $\inf(\pi)$. \CoaX wins $\pi$ iff $\max\{ p(v):v\in \inf(\pi)\}$ is odd. 
\item[Muller.] Is given by subsets $S_1,\ldots, S_n \subseteq V$ of vertices and \CoaX wins a path $\pi$ iff there exists $1 \leq j \leq n$ such that $\inf(\pi) = S_j$.
\end{description}

Following~\cite{AHK02}, we study a multi-player game in which the players are partitioned among two coalitions that play against each other. 

\begin{definition}[Determinacy]
\label{def:det}
Consider a multi-player concurrent game $\G = \zug{m, A,V, \lambda, \delta, \O}$ and an initial vertex $v_0$. Consider a partition $X \cup Y$ of $[m]$. A winning strategy for \CoaX is $\sigma^X$ such that for every strategy $\sigma^Y$ of \CoaY, we have $\path(v_0, \sigma^X, \sigma^Y) \in \O$. Game $\G$ is {\em determined} if for every partition $X \cup Y$ of $[m]$ and every initial vertex $v_0$, either \CoaX has a winning strategy for objective $\O$ or \CoaY has a winning strategy that violates $\O$. 
\end{definition}

It is well-known that turn-based games are determined. 

\begin{theorem}
\cite{Mar75} Two-player turn-based games with M\"uller objectives are determined. 
\end{theorem}

\paragraph*{Determinacy via turn-based games}
We present an alternative way to view determinacy, which is often easier to work with. Consider a multi-player concurrent game $\G = \zug{m, A,V, \lambda, \delta, \O}$ and a partition $X \cup Y$ of $[m]$. Construct two turn-based two-player games $\G_X$ and $\G_Y$ between Players~$X$ and~$Y$. Let $\alpha \in \set{X,Y}$ and denote by $\oalpha$ the other coalition. A player's action in $\G_\alpha$ correspond to a joint action of the corresponding coalition. Intuitively, in game $\G_\alpha$, Player~$\alpha$ is disadvantaged: in each turn, he needs to reveal his choice of action before the other player. Formally, the vertices of $\G_\alpha$ are $V \cup (V \times A^\alpha)$, where the latter are called {\em intermediate vertices}. Player~$\alpha$ controls vertices in $V$ and chooses a joint action $a_\alpha \in A^\alpha$. The game proceeds to a Player~$\oalpha$ controlled vertex $\zug{v, a_\alpha}$, from which she chooses $a_\oalpha$, then the game proceeds to $\delta(v, a_\alpha, a_\oalpha)$. 

Since $\G_X$ and $\G_Y$ are turn-based games, one of the players has a winning strategy. It is not hard to see that Player~$X$ wins $\G_X$ iff \CoaX wins $\G$, and Player~$Y$ wins $\G_Y$ iff \CoaY wins $\G$. Note that since $\G$ is concurrent, it is not necessarily determined, meaning that neither player has a winning strategy, which happens when Player~$Y$ wins $\G_X$ and Player~$X$ wins $\G_Y$. 

\begin{lemma}
\label{lem:TB-det}
A concurrent game $\G$ is determined iff either Player~$X$ wins $\G_X$ or Player~$Y$ wins $\G_Y$. 
\end{lemma}

The lemma provides a scheme to prove determinacy: one needs to show that if Player~$X$ does not win $\G_X$, one needs to show that Player~$Y$ wins $\G_Y$.

\section{Multi-player Discrete-Bidding Games}
We present multi-player discrete-bidding games. An $m$-player discrete-bidding game is played on an arena $\zug{m, k, V, E, \set{T_v}_{v\in V}, \set{R_v}_{v \in V}}$, where $k \in \Nat$ is the total budget, $\zug{V, E}$ is a directed graph, $T_v$ and $R_v$ are permutations over $m$ that specify orders on players that are used, respectively, to break ties and to round payments of bids. We omit $v$ from $T_v$ and $R_v$ when it is clear from the context. A {\em configuration} of a bidding game is $\zug{v, B_1, \ldots, B_m}$ with $B_i \in [k]$ and $\sum_{i \in [m]} B_i = k$, meaning that the token is placed on $v \in V$ and \PLi's budget is $B_i$, for $i \in [m]$. We denote by $\C$, the set of configurations. 

Intuitively, in each turn, each player chooses a bid that does not exceed their available budget, the highest bidder moves the token, and pays his bid to the other players. 
To formalize the semantics of a discrete-bidding game $\G$, we describe the explicit concurrent game $\G' = \zug{A, V', \delta, \lambda}$ that $\G$ corresponds to. Note that both $\G$ and $\G'$ are arenas and we introduce objectives later on. 

The vertices are $\C \cup (\C \times A)$, where we refer to $\C \times \A$ as {\em intermediate vertices}. 
The actions consist of bids and vertices to move to upon winning a bidding, namely $A = [k] \cup V$. 
Biddings occur at configuration vertices. For $c  = \zug{v, B_1, \ldots, B_m} \in \C$ and $i \in [m]$, the available actions for \PLi are bids that do not exceed his budget, thus $\lambda(c, i) = [B_i]$. Consider a joint action $\overline{b} = \zug{b_1, \ldots, b_m} \in [k]^{m}$ at $c$. We identify the player who wins the bidding. Let $b_{\max} =\max \set{b_i}$ be the highest bid. Observe that ties can occur. The winner is the first player according to the order $T_v$ among the players who bid highest, namely $i_{\text{win}} = \max \set{T_v(i): b_i = b_{\max}}$. The budgets are updated as follows. First, $B'_{i_{\text{win}}} = B_{i_{\text{win}}} - b_{\max}$. Second, the bid is paid evenly to the other players, where rounding is determined according to the order $R_v$. Formally, let $b_{\max} = a \cdot (m-1) + r$, for $a,r \in \Nat$. First, for each $i' \neq i_{\text{win}}$, we first pay $a$, thus $B'_{i'} = B_{i'} + a$, and second, an additional coin is given to the top $r$ players according to the order $R_v$. 

We conclude the construction by describing the transitions. From $c$ with joint action $\overline{b}$, the game proceeds to intermediate vertex $\zug{c, \overline{b}}$, which is controlled by Player~$i_{\text{win}}$. His available actions are $\lambda(\zug{c, \overline{b}}, i_{\text{win}}) = \set{u \in V: E(v,u)}$. Choosing action $u \in V$ leads to the configuration vertex $c' = \zug{u, B'_1,\ldots, B'_m}$.

\section{Multi-Player Discrete-Bidding Games are Determined}
\label{sec:determinacy}


In this section we study bidding games that are played between two coalitions. Let $\G = \zug{m, k, X, Y, V, E, \set{T_v}_{v\in V}, \set{R_v}_{v \in V}, \O}$, where $X \cup Y$ is a partition of $[m]$ and $\O$ is an objective for \CoaX. This section is devoted to proving the following result.

\begin{theorem}\label{thm:main}
Muller multi-player discrete-bidding games are determined; namely, for every bidding game $\G$, M\"uller objective $\O$, and configuration $c$, either \CoaX has a strategy from $c$ that guarantees $\O$ or \CoaY has a strategy that violates $\O$. 
\end{theorem}  

In Sec.~\ref{sec:local-global}, we present local determinacy in the context of multi-player bidding games and show how local determinacy implies global determinacy. In Sec.~\ref{sec:prop} we identify properties of bidding games that we use in Sec.~\ref{sec:proof} to show that multi-player bidding games are locally determined and are thus determined. 

\subsection{Local and global determinacy}
\label{sec:local-global}

We describe local determinacy in general games. 
Recall (Lem.~\ref{lem:TB-det}) that in order to prove that a game $\G$ is determined, one needs to show that if Player~$X$ does not win $\G_X$, then Player~$Y$ wins $\G_Y$. Further recall that $\G_X$ is a turn-based game, thus if Player~$X$ does not win $\G_X$, then Player~$Y$ wins $\G_X$. Suppose that this is the case, and let $v$ be a vertex in $\G_X$ from which Player~$Y$ wins. Local determinacy is the following property: from $v$, there is a joint action $a^Y$ that Player~$Y$ can reveal before Player~$X$ such that for all joint actions $a^X$, the next vertex $\delta(v, a^X, a^Y)$ is winning for Player~$Y$. Local determinacy in bidding games is described in terms of a {\em bidding matrix}, which we present next.




\paragraph*{The bidding matrix}
Consider a bidding game $\G$ and a configuration $c = \zug{v, B_1, \ldots, B_m}$. 
The {\em bidding matrix} $M_c$ captures all the possible outcomes of the bidding at $c$. 

\begin{definition}[Single bidder bid]
For $\alpha \in \set{X,Y}$ and $i \in \alpha$, denote by $b^\alpha_i = \zug{0,\ldots,0,b_i,0,\ldots,0} \in [k]^\alpha$ the joint bid in which \PLi bids $b_i \in [k]$ and the other players bid $0$. 
\end{definition}

The following lemma shows that we can restrict the coalitions to choose only single bidder joint bids. 

\begin{lemma}
Consider actions $b^\alpha$ of coalition $\alpha \in \set{X, Y}$. Let $b$ be the highest bid in $b^\alpha$ and $i \in \alpha$ be the player who has precedence out of all players in $\alpha$ that bid $b$. Then, for every action $b^{\oalpha}$ of the other coalition, the outcome of $b^\alpha, b^{\oalpha}$ coincides with the outcome of $b^\alpha_i, b^{\oalpha}$. Thus, if \CoaA has a winning strategy, the coalition has a strategy that only chooses single bidder joint bids. 
\end{lemma}
\stam{
\begin{proof}
If \PLi loses the bidding, the claim is trivial since all bids made by Coalition~$\alpha$ do not affect the outcome. Otherwise, \PLi wins the bidding since in $b^\alpha$, \PLi bids higher than all other players in Coalition~$\alpha$, thus the outcome is the same when all other players in $\alpha$ bid $0$. 
\end{proof}
}

Each row and column in $M_c$ respectively correspond to a single bidder joint bid $b^X_i$ and $b^Y_j$, for $i \in X$ and $j \in Y$ (see Fig.~\ref{fig:matrix}).
Note that $M_c$ has $1+\sum_{i \in X} B_i$ rows and $1+\sum_{i \in Y} B_i$ columns, where the extra row and column are for the joint bid of $0$. An entry $\zug{b^X_i, b^Y_j}$ is either $X$ or $Y$ and this is determined by which coalition wins the intermediate vertex $\zug{c, b^X_i, b^Y_j}$ in $\G_X$\footnote{The choice of $\G_X$ over $\G_Y$ is arbitrary.}.

\paragraph*{From local to global determinacy}
Recall that local determinacy is the property that the winner in a configuration of $\G_X$ has an action that they can reveal first such that no matter how the opponent reacts, the game remains in a winning vertex. 

\begin{definition}[Local determinacy in multi-player bidding games]
A multi-player discrete-bidding game $\G$ is \textbf{locally determined} if for every configuration $c$, the bidding matrix $M_c$ either has a row all of whose entries are $X$ or a column, all of whose entries are $Y$.
\end{definition}

Local determinacy was shown in \cite{AAH21} to lead to (global) determinacy for a class of concurrent games that slightly extend bidding games and which our multi-player discrete-bidding games falls under. We describe the rough idea. Suppose that $\G$ is locally determined, that \CoaX's goal is reachability, and that \CoaY loses $\G^Y$. Thus, \CoaX can force staying in winning vertices in $\G^X$. This does not suffice for winning since progress towards the target needs to be made. Fortunately, the proof works in the other direction works: if \CoaX does not win $\G^X$, then \CoaY, whose objective is safety, can force staying in the winning region of $\G^Y$, which here, suffices for winning. Thus, local implies global determinacy in reachability games. To extend to parity objectives, one reduces a parity game to a reachability game via the ``cycle-forming game'' construction (e.g.,~\cite{AR17}), and reduces M\"uller games to parity games following a standard construction (see~\cite{GTW02}). 


\begin{theorem}\label{localtodet}\cite{AAH19}
If a multi-player M\"uller discrete-bidding game is locally-determined, then it is also (globally) determined. 
\end{theorem}

\subsection{Properties of the bidding matrix}
\label{sec:prop}
In this section, we show properties of bidding matrices with which we will later prove local determinacy. 
The proofs in this section rely on the following lemma, which intuitively states that from a configuration $c$, if the biddings outcome is the same, the winner of the game will be the same, thus the corresponding entries in $M_c$ coincide. We say that two joint bids $\overline{b} = \zug{b_1,\ldots, b_m}$ and $\overline{b'} =  \zug{b'_1,\ldots, b'_m}$ lead to the same outcome if the highest bid is the same, i.e., $b_{\max} =\max \set{b_i} = \max \set{b'_i}$, and the winning player coincides, i.e., $\max \set{T(i): b_i = b_{\max}} = \max \set{T(i): b'_i = b_{\max}}$.

\begin{lemma}
Consider a configuration $c = \zug{v, B_1, \ldots, B_m}$ and two joint bids $\overline{b}$ and $\overline{b'}$ that lead to the same outcome. Then, for $\alpha \in \set{X, Y}$, the winner in $\G_\alpha$ in intermediate vertex $\zug{c, \overline{b}}$ coincides with $\zug{c, \overline{b'}}$. In particular, the corresponding entries in $M_c$ coincide. 
\end{lemma}
\begin{proof}
Observe that $\zug{c, \overline{b}}$ and $\zug{c, \overline{b'}}$ are controlled by the same player, namely the winner of the bidding, denoted $i \in \alpha$, for $\alpha \in \set{X, Y}$. Since the winning bids coincide and the rounding mechanism depends only on $v$, the budget updates coincide. Thus, the neighbors of $\zug{c, \overline{b}}$ and $\zug{c, \overline{b'}}$ coincide. Due to determinacy of turn-based games, each neighbor is either won by $X$ or by $Y$. Then, the winner in $\zug{c, \overline{b}}$ and $\zug{c, \overline{b'}}$ is $\alpha$ iff there is a neighbor that is won by \CoaA. In particular, the winners in both intermediate vertices coincide. 
\end{proof}

\begin{wrapfigure}{r}{0.44\textwidth}
     \centering
     \vspace{-1.5\baselineskip}
     \includegraphics[width=0.33\textwidth]{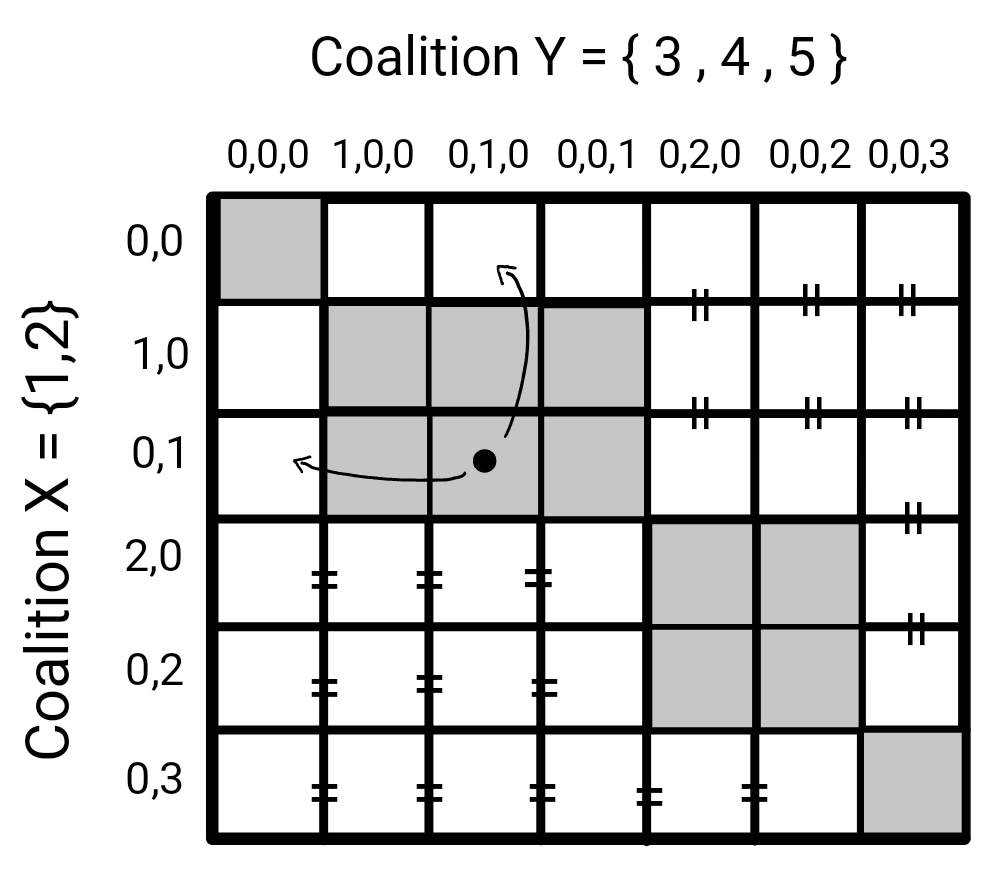}
     \caption{A bidding matrix $M_c$ for $c = \zug{2, 3,1,2,3}$, and $X$=$\{1,2\}$ and $Y$=$\{3,4,5\}$.}
     \label{fig:matrix}
     \vspace{-1.0\baselineskip}
\end{wrapfigure}

We order $M_c$ by highest bid, namely for $b_i < b_j$ with $i,j \in \alpha$, joint bid $b^\alpha_i$ appears before $b^\alpha_j$ (see Fig.~\ref{fig:matrix}). 
For $b \in [k]$, we denote by $S_b$, the {\em tie sub-matrix} that corresponds to outcomes that result in bidding ties with a bid of $b$. In Fig.~\ref{fig:matrix}, $M_c$ has four tie sub-matrices (shaded gray) $S_0$, $S_1$, $S_2$, and $S_3$.
The following lemma is demonstrated in Fig.~\ref{fig:matrix}: the entry represented with a point equals one of the entries that the arrows point to.


\begin{lemma}\label{edges}
    Entry $\zug{b^X_i,b^Y_j}$ in $S_b$ either equals entry $\zug{0,b^Y_j}$ or $\zug{b^X_i,0}$.
\end{lemma}
\begin{proof}
    Suppose entry $\zug{b^X_i,b^Y_j}$ is $X$. Suppose that $T(i)>T(j)$ and the case $T(i)<T(j)$ is dual. Thus, \PLi wins the bidding and pays $b$ to the other players. This coincides with the outcome $\zug{b^X_i,0}$, thus both entries are equal.
\end{proof}

As an example of the following lemma, in Fig.~\ref{fig:matrix}, if $\zug{(0,1),(0,1,0)}$ is not equal to $\zug{(0,0),(0,1,0)}$, then $T(2)>T(4)$.
\begin{lemma}\label{tie}
    Consider an entry $\zug{b^X_i,b^Y_j}$ in $S_b$. If it is not equal to $\zug{0,b^Y_j}$, then $T(i)>T(j)$.
    Similarly, if $\zug{b^X_i,b^Y_j}$ is not equal to $\zug{b^X_i,0}$, then $T(j)>T(i)$.
\end{lemma}
\begin{proof}
    Assume towards contradiction that $T(i)<T(j)$, then Player~$j$ wins the bidding in $\zug{b^X_i,b^Y_j}$, which coincides with the outcome $\zug{0,b^Y_j}$, but then the entries are equal, which is a contradiction to the assumption $\zug{b^X_i,b^Y_j}$ is not equal to $\zug{0,b^Y_j}$. A similar argument shows that if $\zug{b^X_i,b^Y_j}$ is not equal to $\zug{b^X_i,0}$, then $T(j)>T(i)$.
\end{proof}

The following lemma shows that entries above the tie sub-matrix $S_i$ are equal in each column; this is depicted with equal signs in Fig.~\ref{fig:matrix}. 

\begin{lemma}\label{columns}
    Consider a bidding matrix $M_c$. Entries in a column above $S_b$ are equal, 
and entries left of $S_b$ are equal. 
\end{lemma}
\begin{proof}
    Entries above $S_b$ correspond to the bid $b$ of \CoaY and bids of members of \CoaX are $b' < b$. Thus, all the entries corresponds to the same outcome: \CoaY wins the bidding with $b$. The other case is dual.
\end{proof}

\subsection{Multi-player bidding games are locally determined}
\label{sec:proof}
We show that multi-player discrete-bidding games are locally determined, and Thm.~\ref{thm:main} follows from Thm.~\ref{localtodet}.

\begin{proposition}\label{LocallyDet}
Multi-player M\"uller discrete-bidding games are locally determined.
\end{proposition}
\begin{proof}
Consider an $m$-player bidding game $\G$ and a configuration $ c = \zug{v, B_1, \ldots, B_m}$. We show that $M_c$  either has an $X$-row or $Y$-column.

\noindent{\bf Restricting to equal maximal budgets.} 
We prove for $M_c$ in which the richest players in the two coalitions have the same budget, i.e., $\max_{i \in X} B_i = \max_{j \in Y} B_j$. We argue that this implies the claim for general bidding matrices. Indeed, if this is not the case, the claim implies that the top-left sub-matrix $M'_c$ either has an $X$-row or a $Y$-column (see Fig.~\ref{Fig:MaxMutualBudget}; $M'_c$ is the shaded matrix). Assume that $\max_{i \in X} B_i > \max_{j \in Y} B_j$, thus there are rows below $M'_c$, and the other case is dual. There are three cases. First, if $M'_c$ has an $X$-row, this is an $X$-row in $M_c$. Assume $M'_c$ has a $Y$-column (as in the figure). By Lem.~\ref{columns}, the entries in a row below $M'_c$ are all equal. Second, if one of these rows is an $X$-row, it is an $X$-row in $M_c$. Third, all rows below $M'_c$ are $Y$-rows, then together with the $Y$-column in $M'_c$, we find a $Y$-column.

\begin{figure}[!htb]
  \centering 

  \begin{minipage}{0.18\textwidth}
    \centering
    \begin{tikzpicture}[scale=0.49] 
        \def\cols{4}
        \def\rows{6}
        \foreach \x in {1,...,\cols} {
            \foreach \y in {1,...,\rows} {
                \ifnum\y>2
                    \draw[fill=gray!45] (\x,\y) rectangle (\x+1,\y+1);
                    \ifnum\x=3
                        \node at (\x+0.5, \y+0.5) {\scriptsize $Y$};
                    \fi
                \else
                    \draw (\x,\y) rectangle (\x+1,\y+1);
                \fi
            }
        }
        \foreach \rowidx in {1,2} {
            \foreach \x in {1,...,3} { 
                \node at (\x+1, \rowidx+0.5) {\scriptsize $=$};
            }
        }
        \node[left] at (1, 6.5) {\scriptsize 0,0}; 
        \node[left] at (1, 5.5) {\scriptsize 0,1}; 
        \node[left] at (1, 4.5) {\scriptsize 1,0}; 
        \node[left] at (1, 3.5) {\scriptsize 2,0}; 
        \node[left] at (1, 2.5) {\scriptsize 3,0}; 
        \node[left] at (1, 1.5) {\scriptsize 4,0};

        \node[above] at (1.5, 7) {\scriptsize 0,0}; 
        \node[above] at (2.5, 7) {\scriptsize 0,1}; 
        \node[above] at (3.5, 7) {\scriptsize 1,0}; 
        \node[above] at (4.5, 7) {\scriptsize 2,0}; 
    \end{tikzpicture}
    
    \vspace{2mm}
    \caption{The matrix $M_c$ for $c = \zug{4, 1, 2, 1}$, $X = \set{1,2}$, and $Y=\set{3,4}$, has $4$ rows and $2$ columns.}
    \label{Fig:MaxMutualBudget}
  \end{minipage}\hfill 
  %
  \begin{minipage}{0.39\textwidth}
    \centering
    \begin{tikzpicture}[scale=0.45] 
        \foreach \x in {1,...,9} {
            \foreach \y in {1,...,9} {
                \ifnum\y=9 
                    \fill[gray!40] (\x,\y) rectangle (\x+1,\y+1);
                \else
                    \ifnum\x=1 
                        \fill[gray!40] (\x,\y) rectangle (\x+1,\y+1);
                    \else
                        \ifnum\x<6 
                            \ifnum\y>4
                                \fill[gray!40] (\x,\y) rectangle (\x+1,\y+1);
                            \fi
                        \fi
                    \fi
                \fi
            }
        }
        \draw[step=1cm, gray!70, thin] (1,1) grid (10,10);
        \draw[line width=1.0pt] (6,1) rectangle (10,5);

        \foreach \x in {1,...,5} {
            \node at (\x+0.5, 7.5) {\scriptsize $X$};
        }
        \foreach \y in {1,3} {
            \foreach \x in {1,...,5} {
                \node at (\x+0.5, \y+0.5) {\scriptsize $X$};
            }
            \foreach \x in {1,...,4} {
                \node at (\x+1, \y+0.5) {\scriptsize $=$};
            }
        }
        \foreach \x in {6,7,9} {
            \foreach \y in {5,...,9} {
                \node at (\x+0.5, \y+0.5) {\scriptsize $Y$};
            }
            \foreach \y in {5,...,8} {
                \node[rotate=90] at (\x+0.5, \y+1) {\scriptsize $=$};
            }
        }

        \draw[line width=1.2pt] (6,3) rectangle (7,4) node[pos=.5] {\scriptsize $Y$};
        \draw[line width=1.2pt] (7,3) rectangle (8,4) node[pos=.5] {\scriptsize $X$};
        \draw[line width=1.2pt] (9,3) rectangle (10,4) node[pos=.5] {\scriptsize $Y$};
        
        \draw[line width=1.2pt] (6,1) rectangle (7,2) node[pos=.5] {\scriptsize $Y$};
        \draw[line width=1.2pt] (7,1) rectangle (8,2) node[pos=.5] {\scriptsize $X$};
        \draw[line width=1.2pt] (9,1) rectangle (10,2) node[pos=.5] {\scriptsize $X$};

        \node[above] at (1.5, 10) {\scriptsize $\overline{0}$};
        \node[above] at (6.5, 10) {\scriptsize $b^Y_{j_3}$};
        \node[above] at (7.5, 10) {\scriptsize $b^Y_{j_1}$};
        \node[above] at (9.5, 10) {\scriptsize $b^Y_{j_2}$};
        
        \node[left] at (1, 9.5) {\scriptsize $\overline{0}$};
        \node[left] at (1, 7.5) {\scriptsize $r$};
        \node[left] at (1, 4.5) {\scriptsize $b^X_{i_3}$};
        \node[left] at (1, 3.5) {\scriptsize $b^X_{i_1}$};
        \node[left] at (1, 1.5) {\scriptsize $b^X_{i_2}$};

        \draw[-{Stealth[scale=1.0]}, thick] (7.4, 3.8) to[out=160, in=20] node[above, yshift=-0.5mm] {\scriptsize (1)} (1.6, 3.8);
        \draw[-{Stealth[scale=1.0]}, thick] (9.8, 3.5) to[out=30, in=-30, looseness=0.9] node[left, xshift=0.5mm] {\scriptsize (2)} (9.8, 9.5);
        \draw[-{Stealth[scale=1.0]}, thick] (1.5, 1.2) to[out=-40, in=-160] node[above, pos=0.4] {\scriptsize (4)} (7.5, 1.2);
        \draw[-{Stealth[scale=1.0]}, thick] (9.5, 1.75) to[out=50, in=50, looseness=0.3] node[above,pos=0.25] {\scriptsize (3)} (1.8, 1.75);
    \end{tikzpicture}
    \caption{By the induction hypothesis, the shaded matrix has an $X$-row, row $r$, or a $Y$-column. Arrows demonstrate revealing of players.}
    \label{Fig:maintheory}
  \end{minipage}\hfill 
  %
  \begin{minipage}{0.39\textwidth}
    \centering
    \begin{tikzpicture}[scale=0.49] 
        \foreach \x in {1,...,9} {
            \foreach \y in {1,...,9} {
                \ifnum\y=9 
                    \fill[gray!40] (\x,\y) rectangle (\x+1,\y+1);
                \else
                    \ifnum\x=1 
                        \fill[gray!40] (\x,\y) rectangle (\x+1,\y+1);
                    \else
                        \ifnum\x<6 
                            \ifnum\y>4
                                \fill[gray!40] (\x,\y) rectangle (\x+1,\y+1);
                            \fi
                        \fi
                    \fi
                \fi
            }
        }
        \draw[step=1cm, gray!70, thin] (1,1) grid (10,10);
        \draw[line width=1.0pt] (6,1) rectangle (10,5);

        \foreach \x in {1,...,5} {
            \node at (\x+0.5, 7.5) {\scriptsize $X$};
        }
        \foreach \y in {1,3,4} {
            \foreach \x in {1,...,5} {
                \node at (\x+0.5, \y+0.5) {\scriptsize $X$};
            }
            \foreach \x in {1,...,4} {
                \node at (\x+1, \y+0.5) {\scriptsize $=$};
            }
        }
        \foreach \x in {6,7,9} {
            \foreach \y in {5,...,9} {
                \node at (\x+0.5, \y+0.5) {\scriptsize $Y$};
            }
            \foreach \y in {5,...,8} {
                \node[rotate=90] at (\x+0.5, \y+1) {\scriptsize $=$};
            }
        }

        \draw[line width=1.2pt] (6,4) rectangle (7,5) node[pos=.5] {\scriptsize $X$};
        \draw[line width=1.2pt] (7,4) rectangle (8,5) node[pos=.5] {\scriptsize $X$};
        \draw[line width=1.2pt] (9,4) rectangle (10,5) node[pos=.5] {\scriptsize $X$};
        
        \draw[line width=1.2pt] (6,3) rectangle (7,4) node[pos=.5] {\scriptsize $Y$};
        \draw[line width=1.2pt] (7,3) rectangle (8,4) node[pos=.5] {\scriptsize $X$};
        \draw[line width=1.2pt] (9,3) rectangle (10,4) node[pos=.5] {\scriptsize $Y$};
        
        \draw[line width=1.2pt] (6,1) rectangle (7,2) node[pos=.5] {\scriptsize $Y$};
        \draw[line width=1.2pt] (7,1) rectangle (8,2) node[pos=.5] {\scriptsize $\mathbf{X}$};
        \draw[line width=1.2pt] (9,1) rectangle (10,2) node[pos=.5] {\scriptsize $X$};

        \node[above] at (1.5, 10) {\scriptsize $\overline{0}$};
        \node[above] at (6.5, 10) {\scriptsize $b^Y_{j_3}$};
        \node[above] at (7.5, 10) {\scriptsize $b^Y_{j_1}$};
        \node[above] at (9.5, 10) {\scriptsize $b^Y_{j_2}$};

	\node[left] at (1, 9.5) {\scriptsize $\overline{0}$};
        \node[left] at (1, 7.5) {\scriptsize $r$};
        \node[left] at (1, 4.5) {\scriptsize $b^X_{i_3}$};
        \node[left] at (1, 3.5) {\scriptsize $b^X_{i_1}$};
        \node[left] at (1, 1.5) {\scriptsize $b^X_{i_2}$};

        \draw[-{Stealth[scale=1.0]}, thick] (6.2, 4.8) to[out=135, in=45] node[above, yshift=0.1mm] {\scriptsize (1)} (1.5, 4.7);
        \draw[-{Stealth[scale=1.0]}, thick] (1.2, 4.1) to[out=-125, in=-115] node[above,xshift=1.5mm, yshift=-1.5mm] {\scriptsize (2)}  (7.3, 4.2);
        \draw[-{Stealth[scale=1.0]}, thick] (1.2, 4.1) to[out=-135, in=-115, looseness=1.1] node[above, xshift=5.5mm, yshift=-1.5mm] {\scriptsize (3)}  (9.2, 4.2);
    \end{tikzpicture}
    \caption{Adding $i_3$ to $P^X$. Assume $b^Y_{j_3}$ is not a $Y$-column. Arrows (1): revealing the left-most entry, and (2)-(3): every entry $\zug{b^X_{i_3}, j}$, for $j \in P^Y$, is $X$.}
    \label{Fig:nexttheory}
  \end{minipage}

\end{figure}

\smallskip
Recall that an entry in the bidding matrix is $\zug{b^X_i, b^Y_j}$ that represent a pair of joint bids in which the only (potentially) positive bids are $b_i$ and $b_j$. Further recall that a tie matrix $S_b$ corresponds to outcomes that are ties having $b_i = b_j = b$. 

\smallskip
\noindent{\bf Induction.} 
The proof for $M_c$ with equal maximal budgets is by induction on the number of tie matrices in $M_c$. The base case is when $M_c$ contains only $S_0$, namely $M_c$ contains a unique entry $\zug{0,0}$. In other words, all budgets are $0$ and the game is in fact a turn-based game; the player who controls vertex $c$ is determined according to the order $T$. 

\smallskip
\noindent{\bf Inductive step.} 
For the inductive step, assume that $M_c$ has $b$ tie matrices, $S_0, \ldots, S_{b-1}, S_b$, and assume that the top-left matrix $M'_c$ (see Fig.~\ref{Fig:maintheory}; $M'_c$  is shaded) that is obtained from $M_c$ by omitting rows $b^X_i = b^Y_j = b$, either has an $X$-row or a $Y$-column. We prove that $M_c$ either has an $X$-row or a $Y$-column. We prove for the case that $M'_c$ has an $X$-row $r$ (as in the figure), and the case of a $Y$-column is dual. By Lem.~\ref{columns}, the entries in the columns above $S_b$ are all equal. If they are all $X$, then we can extend $r$ to an $X$-row in $M_c$. Thus, consider the case that there is $j_1 \in Y$ such that entry $\zug{r, b^Y_{j_1}}$ is $Y$. 

We maintain two sets of players $P^Y \subseteq Y$ and $P^X \subseteq X$, called {\em revealed players}. 
When $i$ is added to $P^X$, we (1)~reveal that entry $\zug{b^X_i, \overline{0}}$ is $X$, and (2)~\PLi has precedence over all players in $P^Y$, i.e., $T(i) > T(j)$, for all $j \in P^Y$. Dually, when $j$ is added to $P^Y$, (1)~we reveal that entry $\zug{\overline{0}, b^Y_j}$ is $Y$ and (2)~Player~$j$ has precedence over all players in $P^X$. 

Initially, $P^Y = \set{j_1}$ and $P^X = \emptyset$. By the above, entry $\zug{r, b^Y_{j_1}}$ is $Y$, thus by Lem.~\ref{columns} so is entry $\zug{\overline{0}, b^Y_j}$ and since $P^X$ is empty, Player~$j_1$ has precedence trivially. 

\smallskip
\noindent{\bf The first cases.} 
Before continuing to the general case, we find it useful to demonstrate the ideas by describing the first additions to $P^X$ and $P^Y$. 
If Column~$b^Y_{j_1}$ is a $Y$-column, we find a $Y$-column. Otherwise, we argue that there is an entry $\zug{b^X_{i_1}, b^Y_{j_1}}$ in $S_b$ that is $X$ (see Arrow~(1) in Fig.~\ref{Fig:maintheory}). Indeed, by Lem.~\ref{columns}, since $\zug{\overline{0}, b^Y_j}$ is $Y$ so are the other entries in column $b^Y_{j_1}$ above $S_b$. 
By Lem.~\ref{edges}, entry $\zug{b^X_{i_1},b^Y_{j_1}}$ either equals $\zug{b^X_{i_1},\overline{0}}$ or $\zug{\overline{0},b^Y_{j_1}}$, and since the latter is $Y$, the former must be $X$ and $i_1$ is revealed. Moreover, Lem.~\ref{tie} implies that Player~$i_1$ has precedence over Player~$j_1$, i.e., $T(j_1) < T(i_1)$. 
We proceed with a similar argument (see Arrow~(2) in Fig.~\ref{Fig:maintheory}): if $b^X_{i_1}$ is an $X$-row, we are done, otherwise there is a column $b^Y_{j_2}$ with entry $Y$ in $S_b$, and we reveal that entry $\zug{\overline{0}, b^Y_{j_2}}$ is $Y$ and that $j_2$ has precedence over $i_1$. 

The invariant on precedence of players is used when adding $i_2$ to $P^X$. If $b^Y_{j_2}$ is a $Y$-column, we are done. Otherwise, there is an entry $\zug{b^X_{i_2}, b^Y_{j_2}}$ that is $X$. As above (see Arrow~(3) in Fig.~\ref{Fig:maintheory}), entry $\zug{b^X_{i_2}, \overline{0}}$ is $X$ and $i_2$ is revealed. The next step needs to be taken with care. Suppose that row $b^X_{i_2}$ is {\em not} an $X$-row, thus it has a $Y$ entry. The $Y$ cannot appear left of $S_b$. We argue that it is not entry $\zug{b^X_{i_2}, b^Y_{j_1}}$. Indeed (see Arrow~(4) in Fig.~\ref{Fig:maintheory}), we have already established $T(j_1) < T(i_1) < T(j_2) < T(i_2)$, and since $S_b$ corresponds to outcomes that are ties, this means that the outcomes $\zug{b^X_{i_2}, b^Y_{j_1}}$ and $\zug{b^X_{i_2}, b^Y_{j_2}}$ are the same, namely Player~$i_2$ wins. Thus, the two entries coincide in $M_c$ and since the latter is $X$, so is the former. Thus, the $Y$ entry must appear in a column $b^Y_{j_3}$ for $j_3$ that has not yet been revealed.

\smallskip
\noindent{\bf The general case.} 
Let $\hat{j} \in P^Y$ and $\hat{i} \in P^X$ denote the players with the highest precedence according to $T$ (the case of empty $P^X$ has been taken care of above). 
Suppose that $T(\hat{j}) > T(\hat{i})$, then we either find a $Y$-column or add some new $i^*$ to $P^X$. 
A dual argument shows that when $T(\hat{j}) < T(\hat{i})$, then we either find an $X$-row or add a new $j^*$ to $P^Y$. 
Consider the column $b^Y_{\hat{j}}$. If all of its entries are $Y$, we find a $Y$-column. Assume otherwise, that it has an $X$ entry. By Lem.~\ref{columns}, since entry $\zug{\overline{0}, b^Y_{\hat{j}}}$ is $Y$, so are all entries in the column above $S_b$. So, entry $X$ must appear within $S_b$, at row $b^X_{i^*}$. Since $T(\hat{j}) > T(i)$, for all $i \in P^X$ and $S_b$ corresponds to ties, it follows that the outcomes $\zug{b^X_{i}, b^Y_{\hat{j}}}$, for $i \in P^X$, coincide, and so do their entries in $M_c$. Since $\zug{b^X_{\hat{i}}, b^Y_{\hat{j}}}$ is $Y$, so are the rest of the entries. We conclude that exists $i^*$ not in $P^X$ where $\zug{b^X_{i^*}, b^Y_{\hat{j}}}$ is X. By Lem.~\ref{edges}, since $\zug{\overline{0}, b^Y_{\hat{j}}}$ is $Y$, we reveal that $\zug{b^X_{i^*}, \overline{0}}$ is $X$. By Lem.~\ref{tie}, we have $T(i^*) > T(\hat{j})$, and $T(i^*) > T(j)$, for $j \in P^Y$, is obtained inductively. 

\smallskip
\noindent{\bf Conclusion.}
The process above guarantees that in each iteration, if neither an $X$-row nor a $Y$-column is not found, then either a player is revealed from $X$ and is added to $P^X$ or a player is revealed from $Y$ and is added to $P^Y$. Since both $X$ and $Y$ are finite, the process must eventually terminate with either an $X$-row or $Y$-column being found. 
\end{proof}

\begin{remark}[General rounding mechanisms]
Interestingly, the proof of Prop.~\ref{LocallyDet} does not use the linearity of the rounding mechanism and applies for more general mechanisms. We keep the rounding mechanism linear for sake of simplicity. 
\end{remark}

\section{Existence of Nash Equilibrium in Multi-Player Bidding Games}
We introduce non-zero-sum multi-player discrete-bidding games; a game is $\G = \langle m,k,V, E, \set{T_v}_{v \in V},$ $\set{R_v}_{v \in V}, \set{\O_i}_{i \in [m]}\rangle$, where all components are as before except that there are no coalitions rather each player has an individual objective: for $i \in [m]$, the objective of \PLi is $\O_i$. A {\em profile} of strategies is a choice of strategy for each player, i.e., $P = \zug{\sigma_1, \ldots, \sigma_m}$. 
Denote by $\play(c_0, P)$, the play that is generated by $P$ from configuration $c_0$. 
For $i \in [m]$, \PLi's {\em utility} in $P$, denoted $\util_i(c_0, P)$, is $1$ if $\path(c_0, P) \in \O_i$, and otherwise it is $0$.

\begin{definition}[Nash equilibrium]
For a profile $P = \zug{\sigma_1, \ldots, \sigma_m}$ and a strategy $\sigma'_i$, we denote by $P[i \gets \sigma'_i]$, the profile obtained from $P$ by letting \PLi unilaterally switch to strategy $\sigma'_i$, namely $P[i \gets \sigma'_i] = \zug{\sigma_1, \ldots, \sigma_{i-1}, \sigma'_i, \sigma_{i+1}, \ldots, \sigma_m}$. Profile $P$ is a Nash equilibrium (NE) from configuration $c_0$ if for every $i \in [m]$ and every strategy $\sigma'_i$, we have $\util_i(c_0, P) \geq \util_i(c_0, P[i \gets \sigma'_i])$. 
\end{definition}

We show existence of NE. 

\begin{theorem}
\label{thm:NE}
For every non-zero-sum multi-player discrete-bidding game $\G$ with M\"uller objectives and every initial configuration $c_0$, there is a Nash equilibrium from $c_0$. 
\end{theorem}
\begin{proof}
Let $\G = \zug{m,k,V, E, \set{T_v}_{v \in V},\set{R_v}_{v \in V}, \set{\O_i}_{i \in [m]}}$ and initial configuration $c_0$. We describe a profile $P=\zug{\sigma_1,\dots ,\sigma_m}$ and show that it is an NE. Let $i \in [m]$. Define a zero-sum bidding game $\G^i$ by setting $X = \set{i}$ and $Y = [m] \setminus \set{i}$, and $\O = \O_i$. That is, \PLi is playing against a coalition of the rest of the players, each of which abandons their individual objective and aims to violate \PLi's objective. By Thm.~\ref{thm:main}, $\G^i$ is determined. 

We describe $\sigma_i$. Consider a configuration $c$. If $c$ is winning for \PLi in $\G^i$, then $\sigma_i$ follows a winning strategy in $\G^i$. 
Otherwise, recall the explicit concurrent game that corresponds to $\G$; biddings occur in configuration vertices, and once the players choose $\overline{b}$ at $c = \zug{v, B_1, \ldots, B_m}$, the game proceeds to an intermediate vertex $\zug{c, \overline{b}}$ that is controlled by the winner of the bidding. Next, the winner chooses a neighbor $u$ of $v$, the budgets are updated, and the game proceeds to configuration vertex $c' = \zug{u, B'_1, \ldots, B'_m}$. 

Let $I_{\text{win}}(c, i)$ be the set of intermediate vertices such that $\zug{c, \overline{b}}$ is controlled by \PLi, i.e., \PLi wins the bidding at $c$, and $\zug{c, \overline{b}}$ is winning for \PLi in $\G^i$. Intuitively, in a non-winning configuration $c$, \PLi is {\em hopeful}; he chooses the maximal bid such that if he wins the bidding at $c$, he has a winning strategy in the next configuration. 
When $I_{\text{win}}(c, i) = \emptyset$, we define $I_{\text{0-win}}(c, i)$ to be intermediate vertices that are neighbors of $c$ and not controlled by \PLi and are nonetheless winning for \PLi in $\G^i$.

We are ready to define $\sigma_i$. If $I_{\text{win}} = \emptyset$, define $\sigma_i$ to bid $0$, otherwise define $\sigma_i$ to bid $\max\set{b_i:\zug{\zug{b_1,\dots, b_i,\dots,b_m}, c} \in I_{\text{win}}(c, i)}$. 
From an intermediate vertex in $I_{\text{win}} \cup I_{\text{0-win}}$, define $\sigma_i$ to follow a winning strategy in $\G^i$. 
Finally, intuitively, if \PLj, for $j \neq i$, deviates from the action prescribed by $\sigma_j$, we define $\sigma_i$ to join the other players in a coalition that {\em punishes} \PLj; namely, all strategies follow a winning strategy for Coalition~$Y$ in $\G^j$ that violates $\O_j$. 
In App.~\ref{app:NE}, we conclude the formal definition of $\sigma_i$ and prove that $P$ is an NE by showing that if \PLi, for $i \in [m]$, attempts to increase their payoff by unilaterally deviating, the next configuration is losing for \PLi, thus the other players prevent \PLi from increasing the utility. 
\end{proof}

\section{Complexity}
In this section we study the problem of deciding whether a coalition can win an $m$-player game $\G$. 
We note that our lower bound applies for reachability games and in the easier setting that budgets are given in unary, which is in stark contrast to two-player games: parity and mean-payoff discrete-bidding games are in NP and coNP even when the budgets are given in binary~\cite{AS25,AS26}.

\begin{theorem}
\label{thm:PSPACE-hard}
Given a reachability multi-player discrete-bidding game $\G = \langle m,k,X, Y, V, E,$ $ \set{T_v}_{v \in V},\set{R_v}_{v \in V}, \O\rangle$, where $m$ and $k$ are given in unary, two coalitions $X$ and $Y$, and an initial configuration $c_0$, deciding whether \CoaX has a winning strategy is PSPACE-hard.
\end{theorem}
\begin{proof}
We prove NP-hardness by showing a reduction from 3SAT. In App.~\ref{app:PSPACE-hard}, we extend the ideas to a reduction from TQBF and show PSPACE-hardness. 
Consider an input formula $f = C_1 \wedge \ldots \wedge C_p$ over variables $\{x_1,\ldots,x_n\}$. 
We construct an $m=3n+3$ player bidding game $\G$ with a total budget of $k=n+1$ and an initial configuration $c_0$ such that \CoaX wins iff $f$ is satisfiable. 

We associate each variable with three players: for $i \in [n]$, we introduce players $x^0_i, x_i$, and $\neg x_i$. At $c_0$, the budget of Player~$x^0_i$ is $1$ and the other budgets are $0$. So far we have allocated a budget of $n$. We allocate the remaining coin to \PZ. Set $X = \set{0} \cup \set{x^0_i,x_i, x_{\neg i}: i \in [n]}$ and $Y = \set{1,2}$. 
Intuitively, $\G$ consists of two phases. In the first phase, \CoaX chooses an assignment to the variables; technically, they choose who Player~$x^0_i$ pays, Players~$x_i$ or $\neg x_i$, for $i \in [n]$. In the second phase, \CoaY challenges the assignment; they choose a clause $C_j$, and \CoaX wins only if $C_j$ has a literal $\ell$ such that Player~$\ell$'s budget is $1$.

\begin{figure}[htbp]
    \centering
    \begin{tikzpicture}[
        >=Stealth,
        node distance=0.4cm and 1.1cm,
        every node/.style={draw, circle, minimum size=4.5mm, inner sep=0pt, font=\footnotesize},
        every state/.style={draw, circle, minimum size=4.5mm, inner sep=0pt, font=\footnotesize},
        clause/.style={minimum size=3.5mm, font=\tiny}
    ]
    
        \node (v1) {$v_1$};
        \node (v1bar) [below=of v1, yshift=0.15cm] {$\bar{v}_1$};
        \node (v0) [left=of v1, yshift=-0.35cm] {$v_0$};
        
        \node (v2) [right=of v1] {$v_2$};
        \node (v2bar) [right=of v1bar] {$\bar{v}_2$};
        
        \node (v3) [right=of v2] {$v_3$};
        \node (v3bar) [right=of v2bar] {$\bar{v}_3$};
        
        \node (v4) [right=of v3] {$v_4$};
        \node (v4bar) [right=of v3bar] {$\bar{v}_4$};
    
        \node (s_top) [above=0.3cm of v3, xshift=-0.3cm] {$s$};
        \node (s_bot) [below=0.3cm of v3bar, xshift=-0.3cm] {$s$};

        \node (v) [right=of v4, yshift=-0.35cm] {$v$};
    
        \node (C1) [clause, right=of v, yshift=0.45cm] {$C_1$};
        \node (C2) [clause, below=0.15cm of C1] {$C_2$};
        \node (C3) [clause, below=0.15cm of C2] {$C_3$};
    
        \node (t_right) [state, accepting, right=1.3cm of C1, yshift=0.1cm] {$t$};
        \node (s_right) [right=1.3cm of C3, yshift=-0.3cm] {$s$};
    
    
        \draw[->, very thick, red] (v0) -- (v1);
        \draw[->] (v0) -- (v1bar);
    
        \draw[->] (v1) -- (v2);
        \draw[->, very thick, red] (v1) -- (v2bar);
        \draw[->] (v1bar) -- (v2);
        \draw[->] (v1bar) -- (v2bar);
        \draw[->] (v1) -- (s_top);
        \draw[->] (v1bar) to (s_bot);
    
        \draw[->] (v2) -- (v3);
        \draw[->] (v2) -- (v3bar);
        \draw[->, very thick, red] (v2bar) -- (v3);
        \draw[->] (v2bar) -- (v3bar);
        \draw[->] (v2) -- (s_top);
        \draw[->] (v2bar) to (s_bot);
        
        \draw[->, very thick, red] (v3) -- (v4);
        \draw[->] (v3) -- (v4bar);
        \draw[->] (v3bar) -- (v4);
        \draw[->] (v3bar) -- (v4bar);
        \draw[->] (v3) -- (s_top);
        \draw[->] (v3bar) to (s_bot);
    
        \draw[->] (v4) -- (s_top);
        \draw[->] (v4bar) to (s_bot);
        \draw[->] (v4) -- (v);
        \draw[->] (v4bar) -- (v);
    
        \draw[->] (v) to[bend left=42] (t_right);
        \draw[->] (v) -- (C1);
        \draw[->] (v) -- (C2);
        \draw[->] (v) -- (C3);
    
        \draw[->] (C1) -- (t_right);
        \draw[->] (C1) -- (s_right);
    
        \draw[->] (C2) -- (t_right);
        \draw[->] (C2) -- (s_right);
    
        \draw[->] (C3) -- (t_right);
        \draw[->] (C3) -- (s_right);
    
    \end{tikzpicture}
    \caption{The game constructed from the formula $f = C_1 \vee C_2 \vee C_3$ where $C_1 = (x_1 \vee \neg x_2 \vee x_4)$, $C_2 = (\neg x_2 \vee \neg x_3 \vee x_4)$ and $C_3 = (x_1 \vee x_2 \vee \neg x_3)$. Thick red edges depict a \CoaX winning strategy that corresponds to the satisfying assignment $x_1, \neg x_2, x_3, x_4$.}
    \label{fig:reduction}
\end{figure}
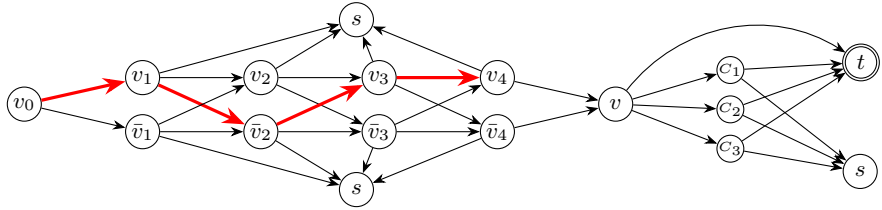

Formally, the vertices of $\G$ are $\set{v_0, v} \cup \set{v_i, v_{\neg i}: i \in [n]} \cup \set{C_j: j\in [p]} \cup \set{t, s}$, where $t$ is \CoaX's target and $s$ is a losing sink. Roughly (see Fig.~\ref{fig:reduction}), $\G$ is a chain that \CoaX must cross to get to $t$. 
We start with the first phase. The initial vertex is $v_0$ and its neighbors are $\set{v_1, v_{\neg 1}, s}$. Both tie and rounding are in favor of \PO. To prevent \PO from moving to $s$, \CoaX must bid $1$ to win the bidding. Since the bid cannot be divided, it is paid to \PO. We construct $\G$ so that in order to have any hope of winning, \CoaX must pay from \PZ's budget. 
For $1 \leq i < n$, the neighbors of both $v_i$ and $v_{\neg i}$ are $\set{v_{i+1}, v_{\neg i+1}, s}$, and the neighbors of $v_n$ and $v_{\neg n}$ are $\set{v, s}$. 
For $i \in [n]$, in both $v_i$ and $v_{\neg i}$, tie are in favor of Player~$x^0_i$ then \PO. Again, if \CoaX hopes to reach $t$, Player~$x^0_i$ must win the bidding with a bid of $1$, otherwise \PO (who has a budget of $1$ at this point) proceeds to $s$. 
Rounding at $v_i$ favors Player~$x_i$ and at $v_{\neg i}$ favors Player~$x_{\neg i}$. 
The key idea: for $1 < i \leq n$, Player~$x^0_{i-1}$ determines the assignment to $x_i$; upon winning the bidding, if he proceeds to $v_i$, Player~$x_i$ will be paid $1$, and if he proceeds to $x_{\neg i}$, Player~$\neg x_i$ is paid $1$. Similarly, \PZ determines the assignment to $x_1$. At $v_n$ and $v_{\neg n}$, Player~$x^0_n$ must win to avoid reaching $s$, and the game proceeds to $v$.

The second phase begins at $v$. As described above, when the game reaches $v$, the budget allocation corresponds to an assignment to $\set{x_1,\ldots, x_n}$, namely, for $i \in [n]$, variable $x_i$ is assigned $1$ or $0$ depending on whose budget is $1$, Player~$x_i$ or $\neg x_i$. 
The neighbors of vertex $v$ are $\set{t, C_1,\ldots, C_p}$, tie breaking favors \PO, and rounding favors \PT. Thus, if \CoaY hopes to win $\G$, they must win the bidding at $v$ by letting \PO bid $1$ and pay the bid to \PT. 
For $j \in [p]$, let clause $C_j$ be $C_j = \ell^j_1 \vee \ell^j_2 \vee \ell^j_3$. The neighbors of vertex $C_j$ are $\set{s, t}$, and tie breaking at vertex $C_j$ favors Players~$\ell^j_1,\ell^j_2$, and $\ell^j_3$, followed by \PT. The key idea: once the game reaches $C_j$, \CoaX wins iff the budget of one of the Players~$\ell^j_1,\ell^j_2$, or $\ell^j_3$ is $1$. Indeed, \PT bids $1$ and will win the bidding and proceed to $s$ if neither of these players bids $1$. On the other hand, if one of them bids $1$, they win and proceed to $t$. Since \CoaY chooses which $C_j$ to move to from $v$, \CoaX wins iff their choice of assignment to the variables satisfies $f$. 
\end{proof}

We conclude this section by describing an upper bound. Recall that a configuration of an $m$-player discrete-bidding game is $\zug{v, B_1, \ldots, B_m}$ such that $v$ is a vertex and $\sum_{1 \leq i \leq m} B_m = k$. Thus, the size of the set of configuration $\C$ is $|V|$ multiplied by the number of ways to divide $k$ coins among $m$ players, which is $k+m-1 \choose m-1$~\cite{Fel68}, which can be bounded by $\min\set{m^k, (k+1)^{m-1}}$. 

\begin{theorem}
Consider an $m$-player discrete-bidding game $\G$ played on a graph with vertex set $V$, a total budget of $k$, and an objective $\O$ in a class $\Gamma$ of objectives. Given an initial configuration $c_0$ and two coalitions $X$ and $Y$, deciding whether \CoaX wins from $c_0$ reduces to solving a turn-based $\Gamma$ game on the configurations $\C$ of $\G$. In particular, when $k$ and $n$ are given in unary, B\"uchi multi-player bidding games are in EXPTIME. 
\end{theorem}

\section{Existence of Values in Mean-Payoff Multi-Player Bidding Games}
Mean-payoff objectives constitute a fundamental quantitative objectives in games~\cite{EM79,ZP96}. In particular, two-player mean-payoff bidding games have been extensively studied; continuous-bidding games exhibit intricate equivalences with a class of stochastic games called {\em random-turn games}~\cite{PSSW09} for a variety of bidding mechanisms~\cite{AHC19,AHI18,AHZ21,AJZ21}, and in discrete-bidding games, recently~\cite{AS26} show an NP and coNP algorithm to decide the optimal payoff even when the game is represented succinctly ($k$ is given in binary). 

We study the basic question on multi-player bidding games; we show existence of a {\em value} in a game between two coalitions. 
A mean-payoff discrete-bidding game is $\G = \zug{m, X, Y, k, V, E, \set{T_v}_{v \in V},\set{R_v}_{v \in V}, w}$, where the components are as in Sec.~\ref{sec:determinacy} and $w: V \rightarrow \Z$ assigns a weight to each vertex. Intuitively, weights are \CoaX's rewards (the maximizing coalition), and their objective is to maximize the long-run average rewards that an infinite play traverses. Formally,

\begin{definition}[Payoffs and value]
For an infinite play $\pi = v_0, v_1, \ldots$, the {\em payoff} is $\MP(\pi) = \lim\inf_{n \to \infty} \frac{1}{n} \sum_{i=1}^{n} w(v_i)$. Existence of a {\em value} is the counterpart of determinacy in quantitative games: for an initial configuration $c_0$, the optimal payoffs that the coalitions can guarantee are $v_X := \sup_{\sigma^X} \inf_{\sigma^Y} \MP(\play(c_0, \sigma^X, \sigma^Y))$ and $v_Y :=  \inf_{\sigma^Y} \sup_{\sigma^X} \MP(\play(c_0, \sigma^X, \sigma^Y))$, then the value exists if $v_X = v_Y$. 
\end{definition}

Our proof of existence of value follows the rough scheme of~\cite{EM79} for turn-based games; first, existence is shown in games played on trees, second, it is extended to general games via a {\em cycle-forming game} construction. We note that both steps require care in multi-player bidding games. Specifically, unlike~\cite{EM79}, our first step does not rely on a value-iteration algorithm, rather we rely on determinacy (Thm.~\ref{thm:main}), and the second step requires reasoning on the configuration graph.

\begin{lemma}
\label{lem:value-in-trees}
A value exists in mean-payoff multi-player bidding games played on a tree. 
\end{lemma}   
\begin{proof}
Consider a game $\G$ between coalitions $X$ and $Y$ that is played on a tree and an initial configuration $c_0$. Since the payoff of a play is prefix independent, we assume that weights appear only in leaves. We describe an algorithm that both finds and establishes existence of a value in $\G$. Let $w_1,\ldots, w_p$ be the unique weights in the game decreasing order, i.e., $w_i > w_{i+1}$, for $1 \leq i < p$. For $1 \leq i \leq p$, let $L_i \subseteq V$ be the set of leaves with weight at least $w_i$, i.e., $\ell \in L_i$ means $w(\ell) \geq w_i$. 
We construct and solve a sequence of reachability bidding games. For $1 \leq i \leq p$, let $\G[i]$ be a reachability bidding game that is obtained from $\G$ by setting \CoaX's target to be $L_i$. Since $\G[i]$ is determined (Thm.~\ref{thm:main}), either \CoaX or \CoaY has a winning strategy from $c_0$. 

Let $1 \leq i \leq p$, be the minimal index such that \CoaX wins $\G[i]$. We claim that $w_i$ is the value of $\G$ from $c_0$. First note that \CoaX can guarantee payoff at least $w_i$; indeed, by following a winning strategy in $\G[i]$, the game is guaranteed to reach a leaf $\ell \in L_i$ with $w(\ell) \geq w_i$. Second, we claim that \CoaY can guarantee payoff at most $w_i$. Note that if $i = 1$, \CoaX can force the highest payoff, thus the claim is trivial since it requires \CoaY to reach some leaf. If $i > 1$, by following a winning strategy $\sigma^Y$ in $\G[i-1]$, \CoaY can force not reaching $L_{i-1}$, in other words, $\sigma^Y$ forces payoff $p < w_{i-1}$. Since every play $\pi$ ends in some leaf, we have $\MP(\pi) = w_j$, for some $j \in [p]$. Thus, if $\MP(\pi) < w_{i-1}$, we have $\MP(\pi) \leq w_i$, and $\sigma^Y$ in fact forces a payoff of at most $w_i$.
\end{proof}

The extension to general games is similar to the extension from trees to general games in turn-based games~\cite{EM79}, but proceeds on the configuration graph of $\G$. The proof can be found in App.~\ref{app:MP-general}. 

\begin{theorem}
\label{thm:MP-general}
Values exist in mean-payoff multi-player discrete-bidding games. 
\end{theorem}   
\section{Discussion and Future Work}
We study, for the first time, multi-player discrete-bidding games. On the positive side, we establish determinacy of qualitative games, a non-trivial property since bidding games are a sub-class of concurrent games, which are not in general determined. While determinacy in two-player discrete-bidding games has been previously established~\cite{DP10,AAH19}, our proof circumvents an inherent obstacle: it is not clear whether or how to extend the conventional reasoning technique in bidding games (reasoning about threshold budgets) to multi-player bidding games. Our negative complexity results serve as an additional witness to the inability to extend this technique to multi-player games. Indeed, in two-player discrete-bidding games, reasoning about thresholds leads to an NP and coNP algorithm; roughly, guess a budget threshold for each vertex and verify the guess~\cite{AS25,AS26}. But we show that multi-player discrete-bidding games are PSPACE-hard. We study, for the first time, infinite-duration non-zero-sum bidding games, and employ determinacy to show existence of Nash equilibrium in pure strategies. Finally, we initiate the study of mean-payoff multi-player bidding games; two-player mean-payoff bidding games have been extensively studied~\cite{AHC19,AHZ21,AHI18,AJZ21}.

For future work, an immediate open question is whether an NE exists in multi-player mean-payoff bidding games. Also, it is interesting to push the boundary of existence of equilibria to stronger equilibria concepts, primarily subgame perfect equilibrium (SPE), which is guaranteed to exist in qualitative turn-based games~\cite{Ume11} and whose complexity has been recently improved~\cite{BRB22}. Taking a step towards mechanism design, auction-based scheduling~\cite{AMS24} requires finding a configuration that is winning for all players, where the total budget could be determined by the algorithm. Studying rational synthesis~\cite{FKL10} in the context of bidding games is also an interesting direction for future work. 



\bibliography{ga}
\appendix

\section{Proof of Thm.~\ref{thm:NE}}
\label{app:NE}
Formally, let $\play(c_0, P) = c_0, b^0, c_1, b^1,\ldots$ be the sequence of configurations and joint bids that are generated if all players follow $P$, where $b^t = \zug{b^t_1,\ldots, b^t_m}$, for $t \geq 0$. Then, if at time $t \geq 0$, for $j \in [m]$, \PLj deviates from his prescribed strategy: at $c_t$ he bids $b'_j \neq b^t_j$ or at intermediate vertex $\zug{b^t, c^t}$ he chooses $c'_{t+1} \neq c_{t+1}$, we will argue below that the next vertices $c'_{t+1}$ and $\zug{b'^t, c_t}$, respectively, are losing for \PLj in $\G^j$. Then, $\sigma_i$ agrees with a winning strategy $\sigma^Y$ for \CoaY in $\G^j$ that violates $\O_j$. 

We claim that $P$ is an NE. Assume towards contradiction that there is $i \in [m]$ and $\sigma'_i$ such that $\util(c_0, P[i \gets \sigma'_i]) > \util_i(c_0, P)$. 
Note that $\util_i(c_0, P) = 0$ and $\util(c_0, P[i \gets \sigma'_i]) =1$. In particular, note that $\play(c_0, P)$ never visits a vertex that is winning for \PLi in $\G^i$ otherwise $\sigma_i$ coincides with a winning strategy and $\util_i(c_0, P) = 1$. 
Observe the last vertex at which $\play(c_0, P[i \gets \sigma'_i])$ differs from $\play(c_0, P)$. We show that after \PLi deviates, the next vertex is losing for \PLi in $\G^i$, thus the other players will form a punishing coalition, which contradicts $\util(c_0, P[i \gets \sigma'_i]) =1$. 
The case that the last shared vertex is an intermediate vertex $\zug{\overline{b}, c}$ is easier. Since $\zug{\overline{b}, c}$ is not winning for \PLi in $\G^i$, every neighbor $c'$ is losing for \PLi in $\G^i$, so a deviation necessarily leads to a losing vertex.  
Next, suppose that the last shared vertex is a configuration vertex $c$. 
Let $b_i$ be $\sigma_i$'s prescribed bid and \PLi deviates to bidding $b'_i$. 
Recall that $\sigma_i$ bids the maximal hopeful bid. Thus, if $b'_i > b_i$, the game necessarily proceeds to an intermediate vertex that is losing for \PLi in $\G^i$. 
The remaining case is $b'_i < b_i$. 
Recall that if $b_i > 0$ and \PLi wins the bidding, the game proceeds to a winning vertex, which we assume is not the case. We conclude by observing that if $\zug{b, c}$ that is not controlled by \PLi is losing, then so is $\zug{b[i \gets b'_i], c}$, for $b'_i < b_i$ since the coalition can follow the same winning strategy. 

\section{Proof of Thm.~\ref{thm:PSPACE-hard}}
\label{app:PSPACE-hard}
Consider an input $\varphi = Q_1 x_1 Q_2 x_2 \ldots Q_n x_n f(x_1, \ldots, x_n)$ to TQBF, where $Q_i \in \set{\exists, \forall}$, and assume wlog that $Q_1 = \exists$ otherwise add a dummy variable. We modify the construction slightly. Let $1 < i \leq n$ such that $Q_i = \forall$. Now, \CoaY determines the assignment to variable $x_i$. Recall that previously, Player~$x^0_{i-1}$ must win the $(i-1)$-th bidding, either at $v_{i-1}$ or $v_{\neg (i-1)}$, and his choice of successor determines the assignment to $x_i$. We make two modification: (1)~Player~$x^0_{i-1}$ belongs to \CoaY and (2)~we set the neighbors of $v_{i-1}$ and $v_{\neg i-1}$ to be $\set{v_i, v_{\neg i}, t}$ so that \CoaY must win the bidding. 

The proof is similar to before. Suppose $\varphi$ is in TQBF and the other direction is dual. Thus, there is a strategy $\tau_\exists$ that assigns truth values to the existentially-quantified variables, such that for every strategy $\tau_\forall$ for the universally-quantified variables, playing $\tau_\exists$ against $\tau_\forall$ leads to a satisfying assignment to $f$. Then, \CoaX wins $\G$ by following $\tau_\exists$: let $1 \leq i \leq n$ such that $Q_i = \exists$ and the game does not reach $s$ or $t$ at time $i-1$, then \CoaY's strategy corresponds to a partial assignment to each $x_j$, with $j < i$ and $Q_j = \forall$. Observe the assignment that $\tau_\exists$ assigns to $x_i$ and define that Player~$x^0_{i-1}$ win the bidding and plays accordingly. Suppose that the second phase is reached, i.e., the game reaches vertex $v$. As before, the configuration corresponds to an assignment to $x_1,\ldots, x_n$, namely for $1 \leq i \leq n$, variable $x_i$ is assigned $1$ if Player~$x_i$'s budget is $1$ and it is assigned $0$ if Player~$x_{\neg i}$'s is $1$. Moreover, the assignment coincides with the assignment generated by playing $\tau_\exists$ against some $\tau_\forall$. Since $\varphi$ is in TQBF, the assignment satisfies $f$. This means that no matter which clause $C_j$ \PO chooses, there is a literal $\ell$ in $C_j$ that is satisfied. The corresponding Player~$\ell$ in \CoaX has a budget of $1$, wins the bidding, and proceeds to $t$.

\section{Proof of Thm.~\ref{thm:MP-general}}
\label{app:MP-general}
Consider a mean-payoff bidding game $\G$ and an initial configuration $c_0$. The {\em cycle-forming game} of $\G$ from $c_0$, denoted $\CFG(\G, c_0)$, is a bidding game that is played on a tree by intuitively simulating $\G$ until a cycle of configurations is closed. Formally, each vertex in $\CFG(\G, c_0)$ corresponds to a history of $\G$. The root vertex is $c_0$. For a history $h \in \C^*$ that ends in $c$, a child of $h$ is $h \cdot c'$, for $c' \in \C$, that is a history of $\G$. That is, $c'$ is a possible successor of the bidding at $c$. Finally, $h = c_0, c_1,\ldots, c_n$ is a leaf of $\CFG(\G, c_0)$ if it is lasso shaped, namely, there is $0 \leq i < n$ such that $c_i = c_n$, and its payoff is average payoff in the cycle, namely $\frac{1}{n-i} \cdot \sum_{i \leq j < n} w(c_j)$. 

By Lem.~\ref{lem:value-in-trees}, $\CFG(\G, c_0)$ has a value $v$. We argue that the value of $\G$ from $c_0$ is $v$. We describe a strategy $\sigma^X$ that guarantees payoff at least $v$ in $\G$, and it is dual to construct a strategy for \CoaY that guarantees payoff at most $v$. Let $\tau^X$ be a \CoaX strategy in $\CFG(\G, c_0)$ that guarantees payoff at least $v$. The strategy $\sigma^X$ simulates $\tau^X$. Consider first a history $h = c_0, c_1, \ldots, c_n$ of $\G$ that does not contain cycles. Then, $h$ is a vertex in $\CFG(\G, c_0)$, and we define $\sigma^X(h) = \tau^X(h)$. If $h$ contains cycles, we omit them greedily as follows. Let $t$ be the first index such that there exists $s < t$ with $c_s = c_t$. Then, construct a history $h' = c_1, \ldots, c_s, c_{t+1}, c_{t+2}, \ldots, c_n$. If $h'$ contains a cycle, set $h = h'$, and repeat. Define $\sigma^X(h) = \tau^X(h')$. Observe that $\sigma^X$ chooses legal bids. Moreover, $h'$ is a history that is consistent with $\tau^X$, thus it guarantees that whenever a cycle is closed, its average payoff is at least $v$. It is not hard to deduce that every play that is consistent with $\sigma^X$ has mean-payoff at least $v$. 

\end{document}